\documentclass[10pt,a4paper]{article}
\usepackage{amssymb}
\usepackage{amsmath}
\usepackage{amsthm}
\usepackage{latexsym}
\usepackage[dvips]{epsfig}
\usepackage{mathrsfs}
\usepackage{yhmath}
\usepackage{bm}

\theoremstyle{plain}
\newtheorem{proposition}{Proposition}
\newtheorem{lemma}{Lemma}
\newtheorem{theorem}{Theorem}

\newtheorem{corollary}{Corollary}

\newcommand{\fixedhat}[1]{\rlap{$\hat{\phantom{#1}}$}#1}

\newcommand{\fixedbreve}[1]{\rlap{$\breve{\phantom{#1}}$}#1}
\newcommand{\fixedhatbreve}[1]{\rlap{$\hat{\breve{\phantom{#1}}}$}#1}

\font\SYM=msbm10
\newcommand{\Real}{\mbox{\SYM R}}

\font\tenscr=rsfs10 scaled1100
\font\sevenscr=rsfs7 
\font\fivescr=rsfs5 
\skewchar\tenscr='177
\skewchar\sevenscr='177
\skewchar\fivescr='177
\newfam\scrfam
\textfont\scrfam=\tenscr
\scriptfont\scrfam=\sevenscr
\scriptscriptfont\scrfam=\fivescr

\begin{document}


\title{\textbf{Approximate twistors and positive mass}}

\author{{\Large Thomas B\"ackdahl} \thanks{E-mail address:
{\tt t.backdahl@qmul.ac.uk}} \\
\vspace{5mm}
{\Large Juan A. Valiente Kroon} \thanks{E-mail address:
{\tt j.a.valiente-kroon@qmul.ac.uk}}\\
School of Mathematical Sciences,\\
 Queen Mary University of London, \\
Mile End Road, London E1 4NS, UK.}

\maketitle

\begin{abstract}
In this paper the problem of comparing initial data to a reference
solution for the vacuum Einstein field equations is considered. This
is not done in a coordinate sense, but through quantification of the
deviation from a specific symmetry. In a recent paper [T. B\"ackdahl,
J.A. Valiente Kroon, Phys. Rev. Lett. \textbf{104}, 231102 (2010)]
this problem was studied with the Kerr solution as a reference
solution. This analysis was based on valence 2 Killing spinors. In
order to better understand this construction, in the present article
we analyse the analogous construction for valence 1 spinors solving
the twistor equation. This yields an invariant that measures how much
the initial data deviates from Minkowski data. Furthermore, we prove
that this invariant vanishes if and only if the mass vanishes. Hence,
we get a proof of the positivity of mass.
\end{abstract}

\section{Introduction}

The idea that constructions involving spinor fields are a valuable
tool for the analysis of global properties of 3-manifolds was brought
to the fore by Witten's proof of the positivity of the mass
\cite{Wit81}. Witten's ideas were then used to prove the positivity of
other notions of mass like Bondi's ---see
e.g. \cite{LudVic81,LudVic82}.

\medskip
In a recent paper \cite{BaeVal10a} spinorial equations have motivated
the construction of a new geometric invariant for initial data sets
for the vacuum Einstein field equations ---see also \cite{BaeVal10b}
for a more detailed exposition. The motivation behind this
construction is to develop a method to compare spacetimes (and initial
data sets thereof) in a gauge-independent, coordinate-free manner. The
key idea is to carry out this comparison by quantifying how much the
spacetime (or its initial data) fail to have a particular
symmetry. Following this strategy, it is possible to obtain a
geometric invariant with the property of vanishing if and only if the
initial data set corresponds to initial data for the Kerr spacetime
---thus, it measures the \emph{non-Kerrness} of the data. 

\medskip
The starting point of the non-Kerrness is the notion of valence 2
Killing spinors. These are spinorial fields
$\kappa_{AB}=\kappa_{(AB)}$ satisfying the equation
\[
\nabla_{A'(A}\kappa_{BC)}=0.
\]
The existence of such spinorial fields in the development of
initial data sets for the vacuum Einstein field equations can be
encoded at the level of the data via the so-called \emph{Killing
spinor initial data equations} ---see
\cite{BaeVal10a,BaeVal10b,GarVal08a}. These equations include among
others the so-called \emph{spatial Killing spinor equation} (an
explanation of the notation is given in the sequel)
\begin{equation}
\label{SpatialKillingSpinorEquation}
\nabla_{(AB} \kappa_{CD)}=0.
\end{equation}
This equation admits, for a generic initial data set for the Einstein
field equations, only the trivial solution. The key insight of
\cite{BaeVal10a,BaeVal10b} was that if one composes the differential
operator of this equation with its formal adjoint one obtains an
elliptic equation which with the appropriate boundary conditions can
be shown to always have a solution. If one evaluates the resulting
spinor in the Killing spinor initial data equations one obtains a
quantitative measure of the deviation from the existence of a symmetry
in the data. The existence of a Killing spinor is a strong
condition to be imposed on a vacuum spacetime. As a consequence, it
turns out that the
construction described in this paragraph can be used to provide a
characterisation of data for the Kerr spacetime.

\medskip
The purpose of the present article is to shed light into the
construction of \cite{BaeVal10a,BaeVal10b} by analysing an analogous
construction motivated by the \emph{twistor equation} (or
\emph{valence 1 Killing spinor equation}):
\[
\nabla_{A'(A}\kappa_{B)}=0. 
\]
In this case, the analogue of equation
\eqref{SpatialKillingSpinorEquation} is the \emph{spatial twistor
equation}
\[
\nabla_{(AB} \kappa_{C)}=0.
\]
This equations has been used to provide conditions on a 3-manifold to
be embeddable in a conformally flat spacetime ---see \cite{Tod84}. The
analogous procedure of \cite{BaeVal10a,BaeVal10b} then produces
an invariant that turns out to be related to the mass.

\medskip
It should be emphasised that in contrast to the analysis of
\cite{BaeVal10a,BaeVal10b} which could be performed, to some greater
length, using tensorial methods, the present discussion is
intrinsically spinorial.

\subsection*{Outline of the article}
Section \ref{Section:Basics} provides a brief discussion of the theory
of spacetimes with solutions to the so-called twistor equation. In
particular, it provides a characterisation of the Minkowski spacetime
in terms of the existence of a specific solution to this
equation. Section~\ref{Section:SpaceSpinors} provides a short overview
of the space spinor formalism to be used in our
analysis. Section~\ref{Section:TwistorData} is concerned with the
question of how to encode in an initial data set that its development
will have a solution to the twistor
equation. Section~\ref{Section:ApproximateTwistor} introduces the
approximate twistor equation: an elliptic equation which with suitable
conditions always admits a unique solution for asymptotically
Euclidean initial data sets ---see Theorem
\ref{Theorem:ExistenceTW}. Whereas Section
\ref{Section:ApproximateTwistor} is concerned with formal elliptic
properties of the equation, Section
\ref{Section:AsymptoticallyEuclideanData} discusses its solvability
for asymptotically Euclidean manifolds. Section
\ref{Section:Invariant} presents a characterisation of Minkowski
initial data by means of a geometric invariant constructed out of the
solution to the approximate twistor equation provided by Theorem
\ref{Theorem:ExistenceTW}. Finally, in Section
\ref{Section:ConnectionMass} we discuss the connection between our
invariant and the mass of the data.

\subsection*{General notation and conventions}
All throughout, $(\mathcal{M},g_{\mu\nu})$ will be an orientable and
time orientable globally hyperbolic vacuum spacetime. Here, and in
what follows, $\mu,\,\nu,\cdots$ denote abstract 4-dimensional tensor
indices. The metric $g_{\mu\nu}$ will be taken to have signature
$(+,-,-,-)$. Let $\nabla_\mu$ denote the Levi-Civita connection of
$g_{\mu\nu}$. The sign of the Riemann tensor will be given by the
equation
\[
\nabla_\mu\nabla_\nu\xi_\zeta-\nabla_\nu\nabla_\mu\xi_\zeta=R_{\nu\mu\zeta}{}^\eta\xi_\eta
\]

\medskip
 The triple $(\mathcal{S}, h_{ab},K_{ab})$ will denote initial data on
a hypersurface of the spacetime $(\mathcal{M},g_{\mu\nu})$. The
symmetric tensors $h_{ab}$, $K_{ab}$ will denote, respectively, the
3-metric and the extrinsic curvature of the 3-manifold
$\mathcal{S}$. The metric $h_{ab}$ will be taken to be negative
definite ---that is, of signature $(-,-,-)$. The indices
$a,\,b,\ldots$ will denote abstract 3-dimensional tensor indices,
while $i,\,j,\ldots$ will denote 3-dimensional tensor coordinate
indices.  Let $D_a$ denote the Levi-Civita covariant derivative of
$h_{ab}$.

\medskip
Spinors will be used systematically. We follow the conventions of
\cite{PenRin84}.  In particular, $A,\,B,\ldots$ will denote abstract
spinorial indices, while $\mathbf{A}, \,\mathbf{B},\ldots$ will be
indices with respect to a specific frame. Tensors and their spinorial
counterparts are related by means of the solder form
$\sigma_\mu{}^{AA'}$ satisfying
$g_{\mu\nu}=\sigma_\mu^{AA'}\sigma_\nu^{BB'}
\epsilon_{AB}\epsilon_{A'B'}$, where $\epsilon_{AB}$ is the
antisymmetric spinor and $\epsilon_{A'B'}$ its complex conjugate
copy. One has, for example, that $\xi_\mu = \sigma_{\mu}{}^{AA'}
\xi_{AA'}$.  Let $\nabla_{AA'}$ denote the spinorial counterpart of
the spacetime connection $\nabla_\mu$. Besides the connection
$\nabla_{AA'}$, two other spinorial connections will be used:
$D_{AB}$, the spinorial counterpart of the Levi-Civita connection
$D_a$ and $\nabla_{AB}$ a Sen connection. Apart from these derivatives we will also use the normal derivative $\nabla \equiv\tau^\mu\nabla_\mu$. Full details will be given in Section~\ref{Section:SpaceSpinors}.

\section{The twistor equation: general theory}
\label{Section:Basics}

A \emph{valence 1 Killing spinor} is a spinor
$\kappa_A$ satisfying the \emph{twistor equation}
\begin{equation}
\label{TwistorEquation} 
\nabla_{A'(A} \kappa_{B)}=0.
\end{equation} 
Taking a further derivative of equation \eqref{TwistorEquation}, antisymmetrising and
commuting the covariant derivatives one finds the integrability condition
\begin{equation}
\Psi_{ABCD}\kappa^{D}=0, \label{IntegrabilityCondition}
\end{equation}
where $\Psi_{ABCD}$ denotes the self-dual Weyl spinor. The above
integrability imposes strong restrictions on the algebraic type of
the Weyl spinor as it states that the vacuum has to be of Petrov type
N or O.

\subsection{A characterisation of the Minkowski spacetime}
We will use valence 1 Killing spinors to formulate a characterisation of the Minkowski spacetime.
\begin{proposition}
\label{Proposition:CharacterisationMinkowski}
Assume that $\kappa_A$ is a solution to $\nabla_{A'(A} \kappa_{B)}=0$
in an asymptotically flat vacuum spacetime. Let
$\eta_A\equiv\nabla_A{}^{A'}\bar\kappa_{A'}$. If $\eta_A\neq 0$ at
some point, then the spacetime is isomorphic to the Minkowski
spacetime.
\end{proposition}
\begin{proof}
As the twistor equation holds, then the integrability condition
\eqref{IntegrabilityCondition} is satisfied.  Similar calculations
show that $\square\kappa_A=0$ and $\nabla_{A'(A}\eta_{B)}=0$.
Therefore we also have $\Psi_{ABCD}\eta^D=0$.  Furthermore we find
$\nabla^A{}_{A'}\eta_A=\tfrac{1}{2}\square\bar\kappa_{A'}=0$. This
means that $\nabla_{AA'}\eta_A=0$. Assume now that $\eta_A\neq 0$ at
some point. Then one has that $\eta_A\neq 0$ everywhere. 

We will now prove that the set $\eta^A\kappa_A=0$ does not have any interior points. We do this by contradiction. Assume $\eta^A\kappa_A=0$ in a neighbourhood of a point $p$. By the relation
$\nabla_{AA'}(\eta^B\kappa_B)=\tfrac{1}{2}\eta_A\bar\eta_{A'}$
we then see that also $\eta_A\bar\eta_{A'}=0$ in a neighbourhood of $p$. This contradicts $\eta_A\neq 0$. Hence, the set $\eta^A\kappa_A=0$ does not have any interior points, thus this set has measure zero.
This means that we can use $\{\eta_A,\kappa_A\}$ locally as a dyad on some neighbourhood around almost every point in the manifold. From the conditions $\Psi_{ABCD}\kappa^D=0$, $\Psi_{ABCD}\eta^D=0$ we conclude that $\Psi_{ABCD}=0$ at almost every point. By continuity we get
$\Psi_{ABCD}=0$ at every point on the manifold. Hence, the spacetime
is conformally flat. Together with asymptotic flatness and the vacuum
field equations, we get that the spacetime is flat, i.e. isomorphic to
the Minkowski spacetime.
\end{proof}

\section{Space spinors: general theory}
\label{Section:SpaceSpinors}
The analysis in this article is based on an analysis of the space
spinor split of equation \eqref{TwistorEquation}. Here we follow the
conventions and notations introduced in
\cite{BaeVal10a,BaeVal10b}. 

\subsection{Basic definitions}
Let $\tau^\nu$ be the future pointing vector tangent to a
congruence of timelike curves and let $\tau^{AA'}$ denote its
spinorial counterpart. We will use the normalization $\tau_{AA'}\tau^{AA'}=2$. Also, let $K_{ABCD}$ denote the spinorial
counterpart of the second fundamental form $K_{ab}$. Furthermore, let
\[
\Omega_{ABCD}\equiv K_{(ABCD)}, \quad K=K_{PQ}{}^{PQ}. 
\]
The Sen connection associated to the congruence defined by
$\tau^{AA'}$ is given by
\[
\nabla_{AB} \equiv \tau_{(A}{}^{A'} \nabla_{B)A'}.
\]
The latter can be written in terms of the intrinsic spinorial Levi-Civita covariant derivative $D_{AB}$ and the spinor $K_{ABCD}$. For example, given a valence 1 spinor $\pi_A$ one has that
\[
\nabla_{AB} \pi_C = D_{AB} \pi_C + \frac{1}{2} K_{ABC}{}^Q\pi_Q,
\]
with the obvious generalisations to higher valence spinors.
Furthermore we denote the normal derivative with
$\nabla\equiv\tau^{AA'}\nabla_{AA'}$. Observe that one can tell the
different derivatives apart by their indices. The spacetime derivative
$\nabla_{AA'}$ has one unprimed index and one primed index, whereas
the Sen connection $\nabla_{AB}$ has two unprimed indices. The normal
derivative $\nabla$ has no indices.

\medskip
Essential for our analysis is the notion of Hermitian conjugation.
Again, given the spinor $\pi_{A}$, we define its \emph{Hermitian
conjugate} via
\[
\hat{\pi}_{A} \equiv \tau_{A}{}^{E'}\bar{\pi}_{E'}.
\] 
The Hermitian conjugate can be extended to higher valence symmetric spinors in the
obvious way.  The symmetric spinors $\nu_{AB}$ and $\xi_{ABCD}$ are said to be
real if 
\[
\hat{\nu}_{AB}=-\nu_{AB},\quad \hat{\zeta}_{ABCD}=\zeta_{ABCD}.
\]
It can be verified that $\nu_{AB} \hat{\nu}^{AB}, \; \zeta_{ABCD}
\hat{\zeta}^{ABCD}\geq 0$. If the spinors are real, then there exist
real spatial tensors $\nu_a$, $\xi_{ab}$ such that $\nu_{AB}$ and
$\xi_{ABCD}$ are their spinorial counterparts. For symmetric spinors
with an odd number of indices like $\kappa_A$, $\xi_{ABC}$ there is no
corresponding notion of reality. However, it can still be shown that $
\kappa_A\hat{\kappa}^A, \; \xi_{ABC}\hat{\xi}^{ABC}\geq 0$. The
differential operator $D_{AB}$ is real in the sense that
\[
\widehat{D_{AB}\pi_C}=-D_{AB}\hat\pi_C.
\]
However,  for the Sen covariant derivative one has that
\[
\widehat{\nabla_{AB}\pi_C}=-\nabla_{AB}\hat\pi_C + \tfrac{1}{2}K_{ABC}{}^D\hat\pi_D.
\]

The restriction to $\mathcal{S}$ of an arbitrary spinor expression
with only unprimed indices can be treated as a spinor expression
intrinsic to $\mathcal{S}$. From Section~\ref{Subsection:IntegrCond} to the end of the paper all expressions will be treated as intrinsic to
$\mathcal{S}$. Before that it will be clear from the context if an
expression is valid only on the slice $\mathcal{S}$ or on the entire
spacetime.

\subsection{Commutators}
Let
\[
\square_{AB} \equiv \nabla_{C'(A}\nabla_{B)}{}^{C'}, \quad
\widehat{\square}_{AB} \equiv \tau_A{}^{A'} \tau_B{}^{B'}\square_{A'B'}=\tau_A{}^{A'} \tau_B{}^{B'}\nabla_{C(A'}\nabla_{B')}{}^{C}.
\]
In vacuum, the action of these operators on a spinor $\pi_A$ is given by
\[
\square_{AB}\pi_C =\Psi_{ABCQ}\pi^Q , \quad \widehat{\square}_{AB} \pi_C
=0.
\]

\medskip
In terms of $\square_{AB}$ and $\widehat{\square}_{AB}$,  the
commutators of $\nabla$ and $\nabla_{AB}$ read
\begin{subequations}
\begin{eqnarray}
&& [\nabla,\nabla_{AB}] ={}
\widehat\square_{AB}-\square_{AB}-\tfrac{1}{2}K_{AB}\nabla+K^D{}_{(A}\nabla_{B)D}-K_{ABCD}\nabla^{CD},\label{commutator1}\\
&&[\nabla_{AB},\nabla_{CD}] ={} \frac{1}{2}\left( 
\epsilon_{A(C}\square_{D)B} + \epsilon_{B(C}\square_{D)A} 
\right) +\frac{1}{2}\left( 
\epsilon_{A(C}\widehat\square_{D)B} 
+\epsilon_{B(C}\widehat\square_{D)A} 
\right)\nonumber \\
&& \hspace{2.5cm}+\frac{1}{2}(K_{CDAB}\nabla-K_{ABCD}\nabla) 
+K_{CDQ(A}\nabla_{B)}{}^Q-K_{ABQ(C}\nabla_{D)}{}^Q. \label{commutator2}
\end{eqnarray}
\end{subequations}

\subsection{Space spinor expressions in Cartesian coordinates}
In the sequel it will be sometimes necessary to give spinorial expressions in terms of Cartesian or asymptotically Cartesian frames and coordinates. For this we make use of the spatial Infeld-van der
Waerden symbols $\sigma^{i}{}_{\mathbf{A}\mathbf{B}}$,
$\sigma_{i}{}^{\mathbf{A}\mathbf{B}}$. Given $x^{i}, \; x_{i}\in
\Real^3$ we shall follow the convention that
\[
x^{\mathbf{AB}} \equiv
\sigma_{i}{}^{\mathbf{AB}} x^{i},
\quad x_{\mathbf{AB}} \equiv
\sigma^{i}{}_{\mathbf{AB}} x_{i},
\]
with
\[
x^{\mathbf{AB}}= \frac{1}{\sqrt{2}}
\left(
\begin{array}{cc}
-x^1 + \mbox{i}x^2 & x^3 \\
x^3 & x^1 +\mbox{i} x^2
\end{array}
\right),
\quad 
x_{\mathbf{AB}}= \frac{1}{\sqrt{2}}
\left(
\begin{array}{cc}
-x_1 - \mbox{i}x_2 & x_3 \\
x_3 & x_1 -\mbox{i}x_2
\end{array}
\right).
\]

\section{Twistor initial data}
\label{Section:TwistorData}
In this section we review some aspects of the space spinor
decomposition of the twistor equation \eqref{TwistorEquation}. A first
analysis along these lines was first carried out in \cite{GarVal08a}. 

\subsection{General observations}
Given a spinor $\kappa_A$ (not necessarily a solution to the twistor
equation), it will be convenient to define the following spinors:
\begin{eqnarray*}
&& \xi_A\equiv \tfrac{2}{3}\nabla_A{}^B\kappa_B,\\
&& \xi_{ABC}\equiv \nabla_{(AB}\kappa_{C)},\\
&& H_{A'AB}\equiv \nabla_{A'(A}\kappa_{B)}.
\end{eqnarray*}
We will use this notation throughout the rest of the paper.
Clearly, for a solution to the twistor equation one has
\[
H_{A'AB}=0. 
\]
The spinors $\xi_A$ and $\xi_{ABC}$ arise in the
space spinor decomposition of the spinor $H_{A'AB}$.
Furthermore, the spinors $\xi_A$ and $\xi_{ABC}$ correspond
to the irreducible components of $\nabla_{AB}\kappa_{C}$ so that one
can write
\begin{equation}\label{SenDiffKappaSplit}
\nabla_{AB}\kappa_C=\xi_{ABC}-\xi_{(A}\epsilon_{B)C}.
\end{equation}
The irreducible components of the
derivative $\nabla_{AB}\xi_C$ in vacuum are given by:
\begin{subequations}
\begin{eqnarray}
&& \nabla_{AB}\xi^B={}
\tfrac{1}{2}\nabla^{BC}\xi_{ABC}
+\tfrac{1}{2}K\xi_A
+\tfrac{1}{2}\Omega_{ABCD}\xi^{BCD},
\label{Senxi1a}\\
&&\nabla_{(AB}\xi_{C)}={}
2\nabla_{(A}{}^D\xi_{BC)D}
+\Psi_{ABCD}\kappa^D
+\tfrac{2}{3}K\xi_{ABC}
-\Omega_{ABCD}\xi^D
-\xi_{(A}{}^{DF}\Omega_{BC)DF}.\label{Senxi1b}
\end{eqnarray}
\end{subequations}


\subsection{Propagation of the twistor equation}

A straightforward consequence of the twistor equation 
\eqref{TwistorEquation} in a vacuum spacetime is that:
\begin{equation}\label{boxkappa}
\square\kappa_A=0,
\end{equation}
where $\square \equiv \nabla^{AA'}\nabla_{AA'}$.  The latter equation
is obtained by applying the differential operator $\nabla^{AA'}$ to
equation \eqref{TwistorEquation} and then using the vacuum commutator
relation for the spacetime Levi-Civita connection. The wave equation
\eqref{boxkappa} plays a role in the discussion of the
\emph{propagation} of the Killing spinor equation ---cfr. \cite{GarVal08a}.

\begin{lemma}
Let $\kappa_A$ be a solution to equation \eqref{boxkappa}. Then the
spinor field $H_{A'AB}$ will satisfy
the wave equation
\begin{equation}
\square H_{A'AB}={}2H_{A'}{}^{CD}\Psi_{ABCD}. \label{wave1} 
\end{equation}
\end{lemma}

The crucial observation is that the right hand side of equation
\eqref{wave1} is a homogeneous expression of the unknown. The hyperbolicity of equation \eqref{wave1} implies the
following result:

\begin{proposition}
\label{Proposition:TwistorDevelopment}
The development $(\mathcal{M},g_{\mu\nu})$ of an
initial data set for the vacuum Einstein field equations,
$(\mathcal{S},h_{ab},K_{ab})$, has a solution to the twistor equation in the 
domain of dependence of $\mathcal{U}\subset\mathcal{S}$ if 
and only if the following equations are satisfied on $\mathcal{U}$.
\begin{subequations} 
\begin{eqnarray}
&& H_{A'AB}=0, \label{old_twd1}\\
&& \nabla H_{A'AB}=0. \label{old_twd2}
\end{eqnarray}
\end{subequations}
\end{proposition}

\subsection{The twistor initial data equations}
\label{Subsection:TwistorInitialDataEqns}

The \emph{twistor initial data} conditions of Proposition~\ref{Proposition:TwistorDevelopment} can be reexpressed in terms of
conditions on the spinor $\kappa_A$ which are intrinsic to the
hypersurface $\mathcal{S}$. Extensive computations using the {\tt
xAct} suite for {\tt Mathematica} render the following result:

\begin{theorem}
\label{Theorem:TwistorData}
Let $(\mathcal{S},h_{ab},K_{ab})$ be an initial data set for the vacuum Einstein field equations, where $\mathcal{S}$ is a Cauchy hypersurface. Let $\mathcal{U}\subset\mathcal{S}$ be an open set. 
The development of the initial data set will then have a solution to
the twistor equation in the domain of dependence of $\mathcal{U}$ if and only if 
\begin{subequations}
\begin{eqnarray}
&& \xi_{ABC}{}=0,\label{twd1}\\
&& \Psi_{ABCD}\kappa^D={}0,\label{twd2}
\end{eqnarray}
\end{subequations}
are satisfied on $\mathcal{U}$. The valence 1 Killing spinor  is obtained by 
evolving \eqref{boxkappa} with initial data satisfying conditions \eqref{twd1}-\eqref{twd2} and
\begin{equation}\label{twd3}
\nabla\kappa_A=-\xi_A,
\end{equation}
on $\mathcal{U}$.
\end{theorem}

\bigskip
\noindent
\textbf{Remark 1.} Conditions \eqref{twd1}-\eqref{twd2} are intrinsic
to $\mathcal{U}\subset \mathcal{S}$ and will be referred to as the
\emph{twistor initial data equations}. In particular, equation
\eqref{twd1}, which can be written as
\begin{equation}
\label{SpatialTwistorEquation}
\nabla_{(AB}\kappa_{C)}=0,
\end{equation}
will be called the \emph{spatial twistor equation}, whereas
\eqref{twd2} will be known as the \emph{algebraic condition}. The
self-dual Weyl spinor $\Psi_{ABCD}$ can be written in terms of
quantities intrinsic to the initial hypersurface $\mathcal{S}$ using
\[
\Psi_{ABCD} = E_{ABCD} + \mbox{i} B_{ABCD},
\]
with
\begin{eqnarray*}
&&  E_{ABCD}= -r_{(ABCD)} + \tfrac{1}{2}\Omega_{(AB}{}^{PQ}\Omega_{CD)PQ}
- \tfrac{1}{6}\Omega_{ABCD}K, \\
&&  B_{ABCD}=-\mbox{i}\ D^Q{}_{(A}\Omega_{BCD)Q},
\end{eqnarray*}
and where the spinor $r_{ABCD}$ is the spinorial counterpart of the
Ricci tensor, $r_{ab}$, of the 3-metric $h_{ab}$.

\begin{proof}
The proof of Theorem \ref{Theorem:TwistorData} consists of a space
spinor decomposition of the conditions
\eqref{old_twd1}-\eqref{old_twd2} and of an analysis of the
dependencies of the resulting conditions.  All calculations are made
on $\mathcal{U}\subset\mathcal{S}$. The equation $H_{A'AB}=0$ is equivalent to 
\begin{eqnarray*}
&& \xi_{ABC}={}0,\\
&& \nabla\kappa_A={}-\xi_A.
\end{eqnarray*}
The wave equation $\square \kappa_A=0$ renders
\[
\nabla\nabla\kappa_A=
-K\nabla\kappa_A
-2\nabla^{BC}\xi_{ABC}
-2\nabla_{AB}\xi^B
+K_{AB}\xi^B
+K^{BC}\xi_{ABC}.
\]
Using the above equations, the equation $\nabla H_{A'AB}=0$ on
$\mathcal{S}$ is seen to be equivalent to 
\begin{eqnarray*}
&&\nabla_{AB}\xi^B={}
\tfrac{1}{2}K\xi_A,\\
&&\nabla_{(AB}\xi_{C)}={}
-\Psi_{ABCD}\kappa^D
-\Omega_{ABCD}\xi^D .
\end{eqnarray*}
Also using the equations \eqref{Senxi1a}-\eqref{Senxi1b}, one will see that 
\begin{eqnarray*}
&& \xi_{ABC}={}0,\\
&& \Psi_{ABCD}\kappa^D={}0,\\
&& \nabla\kappa_A={}-\xi_A,
\end{eqnarray*}
on $\mathcal{U}$ is enough to guarantee $H_{A'AB}=0$ everywhere in the domain of dependence of $\mathcal{U}$, if we evolve $\kappa_A$ by $\square \kappa_A=0$. This completes the proof.
\end{proof}

\subsection{The integrability conditions of the spatial twistor equation}\label{Subsection:IntegrCond}
The condition $\xi_{ABC}\equiv\nabla_{(AB}\kappa_{C)}=0$ 
does not immediately give information about the other irreducible
component of $\nabla_{AB}\kappa_{C}$, namely $\xi_A$. 
However, using $\xi_{ABC}=0$ 
in the relations \eqref{Senxi1a}-\eqref{Senxi1b} one finds that
$\nabla_{AB}\xi_{C}$ can be written in terms of $\xi_A$, $\kappa_A$ and
curvature spinors. We get:
\begin{eqnarray*}
&& \nabla_{AB}\xi^B={}
\tfrac{1}{2}K\xi_A , \\
&&\nabla_{(AB}\xi_{C)}={}
\Psi_{ABCD}\kappa^D
-\Omega_{ABCD}\xi^D .
\end{eqnarray*}

From these we can make the following observation
\begin{lemma}
\label{Lemma:SecondDerivative}
Assume that $\nabla_{(AB}\kappa_{C)}=0$, then
\[
\nabla_{AB}\nabla_{CD}\kappa_E = H_{ABCDE},
\]
where $H_{ABCDE}$ is a linear combination of $\kappa_A$ and
$\nabla_{AB}\kappa_{C}$ with coefficients depending on $\Psi_{ABCD}$, $\hat{\Psi}_{ABCD}$ and $K_{ABCD}$.
\end{lemma}

\medskip
\noindent
\textbf{Remark.} It is important to point out that the assertion of
the Lemma \ref{Lemma:SecondDerivative} is false if $\nabla_{(AB}\kappa_{C)}\neq 0$.

\section{The approximate twistor equation}
\label{Section:ApproximateTwistor}

The spatial twistor equation \eqref{twd1} is an overdetermined
condition for the spinor $\kappa_A$, so we can not expect that a
generic initial data set $(\mathcal{S},h_{ab},K_{ab})$ admits a
solution. One would therefore like to weaken the equation so that it
always admits a unique solution if one specifies the asymptotic
behaviour in a specific way. The strategy will be to compose the spatial twistor operator with its formal adjoint. In this section we will follow this idea and construct the approximate twistor equation. The existence and uniqueness of solutions will be proved in Section~\ref{Section:AsymptoticallyEuclideanData}.

\subsection{The approximate twistor operator}
Let $\mathfrak{S}_1$ and $\mathfrak{S}_3$ denote, respectively, the
spaces of totally symmetric valence 1 and valence 3
spinors. Given $\zeta_{ABC}, \; \chi_{ABC}\in \mathfrak{S}_3$, we
introduce an inner product in $\mathfrak{S}_3$ via:
\[
\langle \zeta_{ABC}, \chi_{DEF}\rangle \equiv \int_{\mathcal{S}} \zeta_{ABC} \hat{\chi}^{ABC} \mbox{d}\mu,
\]
where $\mbox{d}\mu$ denotes the volume form of the 3-metric
$h_{ab}$. We introduce the spatial twistor operator ${\bm \Phi}$
via
\[
{\bm\Phi}: \mathfrak{S}_1 \rightarrow \mathfrak{S}_3, \quad {\bm\Phi}(\kappa)_{ABC}= \nabla_{(AB}\kappa_{C)}. 
\]
Now, consider the pairing
\begin{align*}
\langle \nabla_{(AB}\kappa_{C)}, \zeta_{DEF} \rangle ={}&
\int_{\mathcal{S}} \nabla_{(AB}\kappa_{C)} \hat{\zeta}^{ABC} \mbox{d}\mu 
= \int_{\mathcal{S}} \nabla_{AB}\kappa_{C} \hat{\zeta}^{ABC} \mbox{d}\mu.
\end{align*}
The formal adjoint of the spatial Killing operator, ${\bm\Phi}^*$, can be
obtained from the latter expression by integration by parts. To this
end we note the identity
\begin{align}
\int_\mathcal{U}\nabla_{AB}\kappa_C\hat\zeta^{ABC} \mbox{d}\mu={}&
\int_{\partial\mathcal{U}}n_{AB}\kappa_C\hat\zeta^{ABC} \mbox{d}S
+\int_\mathcal{U}\kappa^C(
\widehat{\Omega_C{}^{ABD}\zeta_{ABD}}
-\widehat{\nabla^{AB}\zeta_{ABC}}
) \mbox{d}\mu, \label{IntegrationbyParts}
\end{align}
with $\mathcal{U}\subset \mathcal{S}$, and where $\mbox{d}S$ denotes
the area element of $\partial \mathcal{U}$, $n_{AB}$ is the spinorial
counterpart of its outward pointing normal, and $\zeta_{ABC}$ is a
symmetric spinor. From \eqref{IntegrationbyParts} it follows that 
\begin{equation}
\label{FormalAdjoint}
{\bm\Phi}^*:\mathfrak{S}_3 \rightarrow \mathfrak{S}_1, \quad {\bm\Phi}^*(\zeta)_A=\nabla^{BC}\zeta_{ABC}-\Omega_A{}^{BCD}\zeta_{BCD}.
\end{equation}
We shall call the composition operator $\mathbf{L}\equiv {\bm\Phi}^*\circ {\bm\Phi}:
\mathfrak{S}_1\rightarrow \mathfrak{S}_1$ given by
\begin{equation}
\mathbf{L}(\kappa_{A}) \equiv \nabla^{BC}\nabla_{(AB}\kappa_{C)}-\Omega_A{}^{BCD}\nabla_{BC}\kappa_D=0, \label{ApproximateTwistorEquation} 
\end{equation}
the \emph{approximate twistor operator}, and equation
\eqref{ApproximateTwistorEquation} the \emph{approximate twistor
  equation}. Note that every solution to the spatial twistor
equation \eqref{SpatialTwistorEquation} is also a solution to
equation \eqref{ApproximateTwistorEquation}. Furthermore, as it will be discussed in the proof of Proposition~\ref{EllipticKernel}, if a solution to the
approximate twistor equation has a sufficiently fast decay at infinity, then it follows from \eqref{IntegrationbyParts} that it is also a solution to the spatial twistor equation. 

\subsection{Ellipticity of the approximate twistor operator}
As a prior step to the analysis of the solutions to the approximate
twistor equation \eqref{ApproximateTwistorEquation}, we look first at
its ellipticity properties.

\begin{lemma}
The operator $\mathbf{L}$ defined by equation
\eqref{ApproximateTwistorEquation} is a formally self-adjoint elliptic
operator.
\end{lemma}

\begin{proof} 
The operator is by construction formally self-adjoint as it is given
by the composition of an operator and its formal adjoint. In order to
verify ellipticity, we will use the fact that commuting derivatives
does not change the principal symbol of an operator. Therefore, we
use the vacuum commutators
to get 
\[
\nabla^{AB}\nabla_{BC}\kappa_A=\tfrac{1}{2}\nabla^{AB}\nabla_{AB}\kappa_C-\tfrac{1}{2}\Omega_{CABD}\nabla^{BD}\kappa^A+\tfrac{1}{3}K\nabla_{CA}\kappa^A .
\]
From this we see that 
\[
\mathbf{L}(\kappa_{A}) \equiv \tfrac{2}{3}\nabla^{BC}\nabla_{BC}\kappa_{A}+\tfrac{2}{3}\Omega_{ABCD}\nabla^{CD}\kappa^B+\tfrac{2}{9}K\nabla_{AB}\kappa^B ,
\]
which is manifestly elliptic.
\end{proof}

\medskip
We note that the approximate twistor
equation \eqref{ApproximateTwistorEquation}  arises naturally from a
variational principle. More precisely, it is the Euler-Lagrange equation of the functional
\begin{equation}
J = \int_{\mathcal{S}} \nabla_{(AB} \kappa_{C)} \widehat{\nabla^{AB} \kappa^{C}} \mbox{d}\mu. \label{functional}
\end{equation}

\section{The approximate twistor equation in asymptotically
  Euclidean manifolds}
\label{Section:AsymptoticallyEuclideanData}

After having studied some formal properties of the twistor
initial data equations \eqref{twd1}-\eqref{twd2},\eqref{twd3}, and the approximate twistor equation 
\eqref{ApproximateTwistorEquation}, we proceed to analyse their
solvability on asymptotically Euclidean manifolds.

\subsection{Asymptotic flatness assumptions}
In what follows, we will be concerned with vacuum spacetimes arising
as the development of asymptotically Euclidean data sets. Let
$(\mathcal{S},h_{ab},K_{ab})$, denote a smooth initial data set for
the vacuum Einstein field equations. For an \emph{asymptotic end} of
$\mathcal{S}$ it will be understood an open set diffeomorphic to the
complement of a closed ball in $\Real^3$.  The 3-manifold $\mathcal{S}$
will be assumed to have an arbitrary number ($N \geq 1$) of
ends. Besides paracompactness and orientability, no further
topological restrictions will be made. Hence, the 3-manifold could
have an arbitrary number of handles.

\medskip
On each asymptotic end it will be assumed that it is possible to introduce asymptotically
Cartesian coordinates $x^i$ with $r=((x^1)^2 +
(x^2)^2 + (x^3)^2)^{1/2}$, such that the intrinsic metric
and extrinsic curvature of $\mathcal{S}$ satisfy in the asymptotic end 
\begin{subequations}
\begin{eqnarray}
&& h_{ij} = -\left(1+ 2m r^{-1}\right)\delta_{ij} + o_\infty(r^{-3/2}), \label{OldDecay1} \\
&& K_{ij} = o_\infty(r^{-5/2}). \label{OldDecay2}
\end{eqnarray}
\end{subequations}
This class of data
can be described as \emph{asymptotically Schwarzschildean}. Here, and
in what follows, the fall off conditions of the various fields will be
expressed in terms of weighted Sobolev spaces $H^s_\beta$, where $s$
is a non-negative integer and $\beta$ is a real number. Here we use
the conventions for these spaces given in \cite{Bar86} ---see also
\cite{BaeVal10b}.  We say that $\eta\in H^\infty_\beta$ if $\eta\in
H^s_\beta$ for all $s$. Thus, the functions in $H^\infty_\beta$ are
smooth over $\mathcal{S}$ and have a fall off at infinity such that
$\partial^l \eta = o(r^{\beta-|l|})$. We will often write
$\eta=o_\infty(r^\beta)$ for $\eta\in H^\infty_\beta$ at the
asymptotic end.

\subsection{Asymptotic form of solutions to the spatial twistor equation}
\label{Section:DecayKappa}
In the sequel, given an initial data set $(\mathcal{S},h_{ab},K_{ab})$
satisfying the decay conditions
\eqref{OldDecay1}-\eqref{OldDecay2}, it will be
necessary to show that it is always possible to solve the equation
\begin{equation}
\label{AsymptoticSpatialTwistorEquation}
\nabla_{(AB}\kappa_{C)}=o_\infty(r^{-3/2}),
\end{equation}
order by order without making any further assumptions on the data. For
this we use asymptotically Cartesian coordinates and consider a
normalised dyad $\{o_A,\iota_A\}$ such that $o_A \iota^A=1$. A direct
calculation allows us to verify the following:

\begin{proposition}
\label{Proposition:AsymptoicSpatialTwistor}
Let $(\mathcal{S},h_{ab},K_{ab})$ denote an initial data set for the
vacuum Einstein field equations satisfying at each asymptotic end the
decay conditions \eqref{OldDecay1}-\eqref{OldDecay2}. Let $m$ denote
the ADM mass of one of these ends. Then on this end
\begin{equation}
\kappa_{\mathbf{A}} = \left(1+\tfrac{1}{2}m r^{-1} \right)x_{\mathbf{AB}} o^{\mathbf{B}} + o_\infty(r^{-1/2}), \label{TwistorLeading}
\end{equation}
satisfies equation \eqref{AsymptoticSpatialTwistorEquation}.
\end{proposition}

\medskip
\noindent
\textbf{Remark.} Formula \eqref{TwistorLeading} implies the
following expansion for $\xi_A$:
\begin{equation}
\xi_{\mathbf{A}}=(1-m r^{-1})o_{\mathbf{A}} +o_\infty(r^{-3/2}).\label{xi1Leading}
\end{equation}
For later reference we notice that
\begin{subequations}
\begin{eqnarray}
&& \nabla_{\mathbf{AB}} r  = - x_{\mathbf{AB}} r^{-1} + o_\infty(r^{-1/2}), \label{xi1LeadingA}\\
&& \nabla_{(\mathbf{AB}} \xi_{\mathbf{C})} = - m r^{-3} x_{(\mathbf{AB}} o_{\mathbf{C})} + o_\infty(r^{-5/2}). \label{xi1LeadingB}  
\end{eqnarray}
\end{subequations}

\subsection{Existence and uniqueness of spinors with twistor asymptotics}
In this section we prove that given a spinor $\kappa_{A}$ satisfying
equations \eqref{AsymptoticSpatialTwistorEquation} and
$\xi_{\mathbf{A}}=o_{\mathbf{A}}+o_\infty(r^{-1/2})$, the asymptotic
expansion \eqref{TwistorLeading} is unique up to a translation.

\begin{proposition}\label{ExistensTwistorAsymptotics}
Given an asymptotic end for which 
\eqref{OldDecay1}-\eqref{OldDecay2} hold, there exists
\begin{equation}
\kappa_{\mathbf{A}}=o_\infty(r^{3/2}),\label{AsymptoticAssumptions1}
\end{equation}
such that
\begin{equation}
\xi_{\mathbf{ABC}}=o_\infty(r^{-3/2}), \quad 
\xi_{\mathbf{A}}=o_{\mathbf{A}}+o_\infty(r^{-1/2}).
\label{AsymptoticAssumptions2}
\end{equation}
The spinor 
$\kappa_{\mathbf{A}}$ is unique up to order $o_\infty(r^{-1/2})$,
apart from a (complex) constant term.
\end{proposition}

\medskip
\noindent
\textbf{Remark 1.} The complex constant term arising in Proposition~\ref{ExistensTwistorAsymptotics}
contains four real parameters. In the sequel, given a particular choice of asymptotically Cartesian coordinates and frame, we will set this constant term to zero. For any other choice of coordinates and frames, the constant can be transformed away by a translation.

\medskip
\noindent
\textbf{Remark 2.} The condition $\xi_{\mathbf{ABC}}=o_\infty(r^{-3/2})$
implies $\xi_{ABC}\in L^2$. Furthermore the conditions in Proposition~\ref{ExistensTwistorAsymptotics} are coordinate independent. 

\begin{proof}
A direct calculation shows that the expansion
\eqref{TwistorLeading} yields \eqref{xi1Leading} and
$\xi_{\mathbf{ABC}}=o_\infty(r^{-3/2})$. Hence,
\eqref{TwistorLeading} gives a solution of the desired form. In
order to prove uniqueness we make use of the linearity of the
integrability conditions \eqref{Senxi1a}-\eqref{Senxi1b}. Note that the translational freedom gives an ambiguity of a constant term in
$\kappa_{\mathbf{A}}$. Let
\begin{equation}
\mathring\kappa_{\mathbf{A}} \equiv 
\left(1+\tfrac{1}{2}m r^{-1} \right)x_{\mathbf{AB}} o^{\mathbf{B}}
\label{TwistorLeadingTerms}
\end{equation}
Let $\breve\kappa_{\mathbf{A}}$, be an arbitrary solution to the system \eqref{AsymptoticAssumptions1}, \eqref{AsymptoticAssumptions2}. Furthermore, let
$\kappa_{\mathbf{A}}\equiv\breve\kappa_{\mathbf{A}}-\mathring\kappa_{\mathbf{A}}$. We
then have
\[
\xi_{ABC}=o_\infty(r^{-3/2}), \quad  \xi_{A}=o_\infty(r^{-1/2}), \quad  \kappa_{A}=o_\infty(r^{3/2}).
\]
 To obtain the desired conclusion we only need to prove that
 $\kappa_{\mathbf{A}}=C_{\mathbf{A}}+o_\infty(r^{-1/2})$, where
 $C_{\mathbf{A}}$ is a constant. This is equivalent to
 $D_{AB}\kappa_{C}=o_\infty(r^{-3/2})$. Note that we now have
 coordinate independent statements to prove.

\medskip
We note that from \eqref{OldDecay1}-\eqref{OldDecay2} it follows that
\[
K_{ABCD}=o_\infty(r^{-5/2}), \quad \Psi_{ABCD}=o_\infty(r^{-3+\varepsilon}),
\]
with $\varepsilon>0$. From \eqref{SenDiffKappaSplit} and
a sharp multiplication result for weighted Sobolev
spaces given in Lemma 2.4 in \cite{Max06} one will get
\begin{align*}
D_{AB}\kappa_{C}={}&\xi_{ABC}-\xi_{(A}\epsilon_{B)C}-\tfrac{1}{2}K_{ABC}{}^D\kappa_D 
= o_\infty(r^{-1/2}).
\end{align*}
Integrating the latter yields 
\[
\kappa_{A}=o_\infty(r^{1/2}).
\]
The constant of integration is incorporated in the remainder
term. Estimating all terms in \eqref{Senxi1a} and
\eqref{Senxi1b} gives
\begin{subequations}
\begin{align}
\nabla_{AB}\xi^B={}&
o_\infty(r^{-5/2}),\label{estSenxi1a}\\
\nabla_{(AB}\xi_{C)}
={}&
o_\infty(r^{-5/2}).\label{estSenxi1b}
\end{align}
\end{subequations}
Hence, $\nabla_{AB}\xi_C=o_\infty(r^{-5/2})$,
and therefore $D_{AB}\xi_C=o_\infty(r^{-5/2})$. 
Integrating this yields $\xi_A=o_\infty(r^{-3/2})$. 
Here the constants of integration are forced 
to vanish by the condition
$\xi_A\negthinspace =\negthinspace o_\infty(r^{-1/2})$.  Hence,
\begin{align*}
D_{AB}\kappa_{C}={}&
\xi_{ABC} 
-\xi_{(A}\epsilon_{B)C} 
-\tfrac{1}{2}K_{ABC}{}^D\kappa_D 
= o_\infty(r^{-3/2}).
\end{align*}
 from where the result follows.
\end{proof}

 From the asymptotic solutions we can obtain a globally defined spinor
$\mathring{\kappa}_{A}$ on $\mathcal{S}$ that will act as a seed for
our approximate twistor.

\begin{corollary}\label{corkapparing}
There are spinors $\mathring{\kappa}_{A}$, defined everywhere on
$\mathcal{S}$, such that the asymptotics at each end is given by
\eqref{TwistorLeading}. Different choices of $\mathring{\kappa}_{A}$ can only differ by a spinor in $H^\infty_{-1/2}$.
\end{corollary}

\begin{proof}
Proposition~\ref{ExistensTwistorAsymptotics} gives the existence at
each end. Smoothly cut off these functions, and paste them
together. This gives a smooth spinor $\mathring{\kappa}_{A}$ defined
everywhere on $\mathcal{S}$. Furthermore
$\nabla_{(AB}\mathring{\kappa}_{C)}\in H^\infty_{-3/2}$.
\end{proof}

\subsection{Fredholm properties}
\label{Section:ApproximateTWinAEM}

In this section we study the invertibility properties of the
approximate twistor operator $\mathbf{L}$ given by equation
\eqref{ApproximateTwistorEquation} on asymptotically Euclidean
manifolds. The necessary elliptic theory for this analysis has been
developed in e.g. \cite{Can81,ChoChr81,ChrOMu81,Loc81}, and has been
adapted to our context in \cite{BaeVal10b}.

\medskip
The decay assumptions \eqref{OldDecay1}-\eqref{OldDecay2} imply that
$\mathbf{L}$ is an \emph{asymptotically homogeneous elliptic operator}
---see e.g. \cite{Can81,Loc81}.  This is the standard assumption on
elliptic operators on asymptotically Euclidean manifolds. It follows from \cite{Can81}, Theorem 6.3 that:

\begin{lemma}
The elliptic operator 
\[
\mathbf{L}: H^{2}_{\delta} \rightarrow H^{0}_{\delta-2}, 
\]
where $\delta$ is not a non-negative integer is a linear operator with
finite dimensional Kernel and closed range.
\end{lemma}

We will also need the following ancillary result ---cfr. \cite{ChrOMu81} for an analogous result for Killing vectors.

\begin{proposition}
\label{Proposition:TrivialityKernel}
Let $\nu_{A}\in H^\infty_{-1/2}$ such that
$\nabla_{(AB}\nu_{C)}=0$. Then $\nu_{A}=0$ on $\mathcal{S}$.
\end{proposition}

\begin{proof}
We will use Theorem 20 of \cite{BaeVal10b}, which is an adaptation of a
result in \cite{ChrOMu81}. From Lemma \ref{Lemma:SecondDerivative} it
follows that $\nabla_{AB} \nabla_{CD} \nu_{E}$ can be expressed as a
linear combination of lower order derivatives with smooth coefficients
with the proper decay. Thus, Theorem 20 of \cite{BaeVal10b} applies
with $m=1$ and one obtains the desired result.
\end{proof}

\medskip
We are now in the position to discuss the Kernel of the approximate
twistor operator in the case of spinor fields that decay at infinity.

\begin{proposition}\label{EllipticKernel}
Let $\nu_{A}\in H^\infty_{-1/2}$. If $\mathbf{L}(\nu_{A})=0$, then
$\nu_{A}=0$.
\end{proposition}

\begin{proof}
Using the identity \eqref{IntegrationbyParts} with
$\zeta_{ABC}= \nabla_{(AB} \nu_{C)}$ and assuming that
$\mathbf{L}(\nu_{C})=0$, one obtains
\[
\int_{\mathcal{S}} \nabla^{AB}\nu^{C} \widehat{\nabla_{(AB}
\nu_{C)}} \mbox{d}\mu = \int_{\partial\mathcal{S}_\infty}
n^{AB}\nu^{C} \widehat{\nabla_{(AB}\nu_{C)}} \mbox{d}S,
\]
where $\partial S_\infty$ denotes the sphere at infinity. Assume now,
that $\nu_{A}\in H^\infty_{-1/2}$. It follows that $\nabla_{(AB}
\nu_{C)}\in H^\infty_{-3/2}$ and furthermore, using the finer
multiplication Lemma 15 of \cite{BaeVal10b} that
\[
n^{AB} \nu^{C}
\widehat{\nabla_{(AB}\nu_{C)}} = o(r^{-2}).
\]
 The integration of the
latter over a finite sphere of sufficiently large radius is of type
$o(1)$. Thus one has that
\[
\int_{\partial\mathcal{S}_\infty}
n^{AB}\nu^{C} \widehat{\nabla_{(AB}\nu_{C)}} \mbox{d}S=0,
\]
 from where
\[
\int_{\mathcal{S}} \nabla^{AB}\nu^{C} \widehat{\nabla_{(AB} \nu_{C)}} \mbox{d}\mu =0.
\]
Therefore, one concludes that 
\[
\nabla_{(AB} \nu_{C)}=0.
\]
That is, $\nu_{A}$ has to be a solution to the spatial twistor
equation. Using Proposition~\ref{Proposition:TrivialityKernel} it
follows that $\nu_{A}= 0$ on $\mathcal{S}$.
\end{proof}

\subsection{Existence of approximate twistors}
We are now in the position of providing an existence proof to
solutions to equation \eqref{ApproximateTwistorEquation} with the
asymptotic behaviour discussed in section \ref{Section:DecayKappa}.

\begin{theorem}
\label{Theorem:ExistenceTW}
Given an asymptotically Euclidean initial data set
$(\mathcal{S},h_{ab},K_{ab})$ satisfying the asymptotic conditions
\eqref{OldDecay1}-\eqref{OldDecay2}, there exists a smooth unique
solution to equation \eqref{ApproximateTwistorEquation} with
asymptotic behaviour at each end given by \eqref{TwistorLeading}.
\end{theorem}

\begin{proof} 
We consider the Ansatz
\[
\kappa_{A} = \mathring{\kappa}_{A} + \theta_{A}, \quad \theta_{A} \in H^2_{-1/2},
\]
with $\mathring{\kappa}$ given by Corollary
\ref{corkapparing}. Substitution into equation
\eqref{ApproximateTwistorEquation} renders the following equation for
the spinor $\theta_{A}$:
\begin{equation}
\label{elliptic:general}
\mathbf{L}(\theta_{C}) = -\mathbf{L}(\mathring{\kappa}_{C}).
\end{equation}
By construction it follows that $\nabla_{(AB}
\mathring{\kappa}_{C)}\in H^\infty_{-3/2}$, so that
$F_{C}\equiv-\mathbf{L}(\mathring{\kappa}_{C})\in H^\infty_{-5/2}$.
Using the Fredholm Alternative for second order elliptic systems
(cfr. Theorem 23 in \cite{BaeVal10b}), one concludes that equation
\eqref{elliptic:general} has a unique solution if $F_{A}$ is
orthogonal to all $\nu_{A}\in H^0_{-1/2}$ in the Kernel of
$\mathbf{L}^*=\mathbf{L}$. Proposition~\ref{EllipticKernel} states that this Kernel is
trivial. Thus, there are no restrictions on $F_{A}$ and equation
\eqref{elliptic:general} has a unique solution as desired. Due to
elliptic regularity, any $H^2_{-1/2}$ solution to the previous
equation is in fact a $H^\infty_{-1/2}$ solution ---cfr. Theorem
24 in \cite{BaeVal10b}. Thus, $\theta_{A}$ is
smooth.  To see that $\kappa_{A}$ does not depend on the particular
choice of $\mathring\kappa_{A}$, let $\mathring\kappa'_{A}$, be
another choice. Let $\kappa'_{A}$ be the corresponding solution to
\eqref{elliptic:general}. Due to Corollary \ref{corkapparing}, we have
$\mathring\kappa_{A}-\mathring\kappa'_{A} \in H^\infty_{-1/2}$.
Hence, we have $\kappa_{A}-\kappa'_{A}\in H^\infty_{-1/2}$ and
$\mathbf{L}(\kappa_{A}-\kappa'_{A})=0$. According to Proposition~\ref{EllipticKernel}, $\kappa_{A}-\kappa'_{A}=0$, and the proof is
complete.
\end{proof}

\medskip
The following is a direct consequence of Theorem
\ref{Theorem:ExistenceTW}, and will be crucial for the construction of our geometric invariant.

\begin{corollary}
\label{Corollary:Boundedness}
A solution, $\kappa_{A}$, to equation
\eqref{ApproximateTwistorEquation} with asymptotic behaviour given by
\eqref{TwistorLeading} satisfies $J<\infty$ where $J$ is the
functional given by equation \eqref{functional}.
\end{corollary}

\begin{proof}
The functional $J$ given by equation
\eqref{functional} is the $L^2$ norm of
$\nabla_{(AB}\kappa_{C)}$. Now, if $\kappa_{A}$ is the solution
given by Theorem \ref{Theorem:ExistenceTW}, one has that
$\nabla_{(AB}\kappa_{C)}\in H^0_{-3/2}$. In our conventions this reads 
\[
J=\Vert\nabla_{(AB}\kappa_{C)}\Vert_{L^2} =\Vert\nabla_{(AB}\kappa_{C)}\Vert_{H^0_{-3/2}}<\infty.
\]
The result follows. 
\end{proof}

\medskip
\noindent
\textbf{Remark.} Again, let $\kappa_{A}$ be the solution to equation
\eqref{ApproximateTwistorEquation} given by Theorem
\ref{Theorem:ExistenceTW}.  Using the identity
\eqref{IntegrationbyParts} with $\zeta_{ABC}=\nabla_{(AB}\kappa_{C)}$
one obtains
\[
J = \int_{\partial \mathcal{S}_\infty} n^{AB}\kappa^{C}
\widehat{\nabla_{(AB}\kappa_{C)}} \mbox{d}S <\infty.
\]
Thus, the invariant $J$ evaluated at the solution $\kappa_A$ given
by Theorem \ref{Theorem:ExistenceTW} can be expressed as a boundary
integral at infinity. A crude estimation of the integrand of the
boundary integral does not allow us to directly establish its
boundedness. In any case this follows from Corollary~\ref{Corollary:Boundedness}.

\section{The geometric invariant}
\label{Section:Invariant}
In this section we use the functional given by \eqref{functional}
and the algebraic condition \eqref{twd2}  to
construct the desired geometric invariant measuring the deviation of
$(\mathcal{S},h_{ab},K_{ab})$ from Minkowski initial data. For this purpose, let $\kappa_A$ be a solution to equation
\eqref{ApproximateTwistorEquation} as given by Theorem
\ref{Theorem:ExistenceTW}. Define
\begin{eqnarray}
&& I' \equiv \int_{\mathcal{S}} \Psi_{ABCD}\kappa^D
\hat{\Psi}^{ABC}{}_F\hat{\kappa}^F \mbox{d}\mu . \label{I1} 
\end{eqnarray}
The geometric invariant is then defined by
\begin{eqnarray}
I \equiv J + I'. \label{geometric:invariant}
\end{eqnarray}

\medskip
\noindent
\textbf{Remark.} It should be stressed that by construction $I$ is
coordinate independent and that $I\geq 0$. We also have the following
lemma.

\begin{lemma}
The geometric invariant given by \eqref{geometric:invariant} is finite
for an initial data set $(\mathcal{S},h_{ab},K_{ab})$ satisfying the
decay conditions \eqref{OldDecay1}-\eqref{OldDecay2}.
\end{lemma} 

\begin{proof}
 From Corollary \ref{Corollary:Boundedness} we already have
$J<\infty$. From the form of the decay assumptions
\eqref{OldDecay1}-\eqref{OldDecay2} we have $\Psi_{ABCD}\in
H^\infty_{-3+\varepsilon}$, $\varepsilon>0$.  Using the multiplication
rule for weighted Sobolev spaces ---see e.g. Theorem 23 in
\cite{BaeVal10b}--- together with $\kappa_{A}\in
H^\infty_{1+\varepsilon}$ we obtain 
\[
\Psi_{ABCD}\kappa^D \in H^\infty_{-3/2}.
\]
 Thus, again one finds that $I'<\infty$. Hence, the invariant \eqref{geometric:invariant} is finite and well defined.
\end{proof}

\medskip
The invariant $I$ can be used to provide a global characterisation of
Minkowski initial data. More precisely, one has that:

\begin{theorem}
\label{Theorem:Characterisation}
Let $(\mathcal{S},h_{ab},K_{ab})$ be an asymptotically Euclidean
initial data set for the vacuum Einstein field equations satisfying on
each of its asymptotic ends the decay conditions
\eqref{OldDecay1}-\eqref{OldDecay2}. Let $I$ be the invariant defined
by equations \eqref{functional}, \eqref{I1} and
\eqref{geometric:invariant}, where $\kappa_{A}$ is given as the only
solution to equation \eqref{ApproximateTwistorEquation} with
asymptotic behaviour on each end given by \eqref{TwistorLeading}. The
invariant $I$ vanishes if and only if $(\mathcal{S},h_{ab},K_{ab})$ is
an initial data set for the Minkowski spacetime.
\end{theorem}

\begin{proof}
Due to our smoothness assumptions, if $I=0$ it follows that equations
\eqref{twd1}-\eqref{twd2} are satisfied on the whole of
$\mathcal{S}$. Thus, the development of $(\mathcal{S},h_{ab},K_{ab})$
will have, at least in a slab, a solution to the twistor equation
\eqref{TwistorEquation}. Now, because of the asymptotic behaviour
\eqref{TwistorLeading} one can find points in the development of the
data for which $\eta_A = \nabla_A{}^{A'} \bar{\kappa}_{A'}\neq
0$. Thus, in view of Proposition~\ref{Proposition:CharacterisationMinkowski} one has that the
development of the data is isomorphic to the Minkowski spacetime.
\end{proof}

\section{Connection to the mass}
\label{Section:ConnectionMass}

As a consequence of the Theorem of the Positivity of the Mass
\cite{SchYau79,SchYau81a,Wit81} one knows that if the ADM mass of a
regular initial data set for the vacuum Einstein field equations
$(\mathcal{S},h_{ab},K_{ab})$ vanishes, then the initial data must be
data for the Minkowski spacetime. As a consequence, the mass provides
a characterisation of Minkowskian data. This suggests that the
geometric invariant given by Theorem~\ref{Theorem:Characterisation}
must be related to the mass of the initial data set. In this section
we show that this is indeed the case. More precisely, our methods
provide a proof of the following positivity of mass result:

\begin{theorem}
Let $(\mathcal{S},h_{ab},K_{ab})$ be an initial data set for the
vacuum Einstein field equations satisfying the decay conditions
\eqref{OldDecay1}-\eqref{OldDecay2}. Furthermore assume that $\mathcal{S}$ has the topology of $\mathbb{R}^3$. Let $m$ be the ADM
mass of the asymptotic end of the initial data set. Then $m$ is
non-negative. Moreover, if $m=0$, then $(\mathcal{S},h_{ab},K_{ab})$
is initial data for the Minkowski spacetime.
\end{theorem}

\medskip
The proof of this theorem, together with some other relevant
observations, will be given in the following subsections.

\subsection{An expression for the mass}
The following analysis is valid for a slice $\mathcal{S}$ with an arbitrary number of asymptotic ends, without any extra restrictions on the topology.
Let $\kappa_A$ be a solution to $\mathbf{L}(\kappa_A)=0$. A
calculation reveals that
\begin{equation}
\mathbf{L}(\xi_A)=\;-\tfrac{2}{3}\xi^B\nabla_{AB}K.
\label{Equation:xiA}
\end{equation}
In general, if $\zeta_A$ is an arbitrary spinor, one has that
\begin{eqnarray*}
&&\widehat{\mathbf{L}(\zeta_A)}=-\nabla^{BC}\widehat{\nabla_{(AB}\zeta_{C)}}\\
&& \phantom{\widehat{\mathbf{L}(\zeta_A)}}=\mathbf{L}(\hat\zeta_A)+\tfrac{2}{3}\zeta^B\nabla_{AB}K.
\end{eqnarray*}
In particular, for $\xi_A$ as given by equation \eqref{Equation:xiA}
one has that that 
\begin{equation}
\mathbf{L}(\hat\xi_A)=0.
\label{Observation1}
\end{equation}
 We will exploit this observation to obtain an alternative expression
for the total mass of an initial data set $(\mathcal{S},h_{ab},K_{ab})$. Let
\begin{equation}
\label{KprimDef}
M'\equiv{}
\int_\mathcal{S}\nabla_{(AB}\hat\xi_{C)}\widehat{\nabla^{(AB}\fixedhat\xi^{C)}} \mbox{d}\mu.
\end{equation}
Now, under the decay assumptions \eqref{OldDecay1}-\eqref{OldDecay2},
the expansions \eqref{xi1Leading},
\eqref{xi1LeadingA}-\eqref{xi1LeadingB} and equation
\eqref{Observation1} render
\begin{eqnarray*}
&& M'={}
\int_{\partial\mathcal{S}_\infty}n_{AB}\hat\xi_{C}\widehat{\nabla^{(AB}\fixedhat\xi^{C)}} \mbox{d}S
+\int_\mathcal{S}\hat\xi_{A}\widehat{\mathbf{L}(\fixedhat\xi^{A})}\mbox{d}\mu\nonumber\\
&& \phantom{M'}={}
\int_{\partial\mathcal{S}_\infty}n_{AB}\hat\xi_{C}\widehat{\nabla^{(AB}\fixedhat\xi^{C)}} \mbox{d}S\nonumber\\
&& \phantom{M'}={} 4\pi M,
\end{eqnarray*}
where $M$ is the sum of the ADM masses of the asymptotic ends ---the total mass. Notice that the integral $M'$ is the $L^2$ norm of
$\nabla_{(AB}\hat\xi_{C)}$ and hence $M'\geq 0$. That is $M\geq 0$.

\subsection{An alternative expression for the total mass}
The Hermitian conjugate of the symmetrized derivative of $\kappa_A$ is
\begin{equation}
\nabla_{(AB}\hat\kappa_{C)}=-\hat\xi_{ABC}-\Omega_{ABCD}\hat\kappa^D.
\end{equation}
Furthermore
\begin{equation}
\mathbf{L}(\hat\kappa_A)=-\tfrac{2}{3}\hat\kappa^B\nabla_{AB}K.
\end{equation}
We would like to conclude that $\mathbf{L}(\hat\kappa_A)\in
H^0_{-5/2}$. For that purpose we study the function
$f\equiv\kappa_A\hat\kappa^A \sigma^{-2}$. Due to the specified
asymptotics, this function is bounded at infinity. Furtheremore it is
continous, and therefore bounded everywhere on $\mathcal{S}$.  Let
$C\equiv \sup_{\mathcal{S}} f < \infty$, then we have
\begin{align}
\Vert\mathbf{L}(\hat\kappa_A)\Vert^2_{H^0_{-5/2}}&=
\tfrac{4}{9}\Vert\hat\kappa^B\nabla_{AB}K\Vert^2_{H^0_{-5/2}}
=\tfrac{2}{9}\int_{\mathcal{S}}\kappa_A\hat\kappa^A D_{BC}K \widehat{D^{BC}K} \sigma^2 \mbox{d}\mu\nonumber\\
&\leq 
\tfrac{2}{9}C\int_{\mathcal{S}}D_{BC}K \widehat{D^{BC}K} \sigma^4\mbox{d}\mu
=\tfrac{2}{9}C\Vert D_{AB}K\Vert^2_{H^0_{-7/2}}
<\infty.
\end{align}
Let $\theta_A\in H^2_{-1/2}$ be the unique solution to the
elliptic equation
\[
\mathbf{L}(\theta_A)=\mathbf{L}(\hat\kappa_A).
\]
The existence and uniqueness is guaranteed by $\mathbf{L}(\hat\kappa_A)\in H^0_{-5/2}$ and the Fredholm alternative ---see the proof of Theorem~\ref{Theorem:ExistenceTW}. 
Motivated by the previous equation one defines
$\breve\kappa_A\equiv\theta_A-\hat\kappa_A$. Clearly, one has that
$\mathbf{L}(\breve\kappa_A)={}0$. By elliptic regularity we have $\breve\kappa_A\in H^\infty_{3/2}$. Let
\begin{eqnarray*}
&& \breve\xi_A\equiv{}\tfrac{2}{3}\nabla_A{}^B\breve\kappa_B. 
\end{eqnarray*}
Some further computations reveal that
\begin{eqnarray*}
&& \breve\xi_A={}\tfrac{2}{3}\nabla_A{}^B\theta_B
+\hat\xi_A+\tfrac{1}{3}K \hat\kappa_A,\\
&& \mathbf{L}(\breve\xi_A)={}-\tfrac{2}{3}\breve\xi^B\nabla_{AB}K,\\
&& \nabla^{BC}(\Omega_{ABCD}\breve\xi^D)={}
\Omega_{ABCD}\nabla^{BC}\breve\xi^D
-\tfrac{1}{2}\breve\xi_A\Omega_{BCDF}\Omega^{BCDF}
+\tfrac{2}{3}\breve\xi^B\nabla_{AB}K .
\end{eqnarray*}
Hence, 
\[
\nabla^{BC}(\nabla_{(AB}\breve\xi_{C)}+\Omega_{ABCD}\breve\xi^D)=
\Omega_{ABCD}\nabla^{BC}\breve\xi^D
-\tfrac{1}{2}\breve\xi_A\Omega_{BCDF}\Omega^{BCDF}.
\]
Furthermore,  it can be seen that
\[
\widehat{\nabla_{(AB}\fixedhatbreve{\xi}{}_{C)}}=\nabla_{(AB}\breve\xi_{C)}
+\Omega_{ABCD}\breve\xi^D.
\]
Define
\begin{equation}
M''\equiv{}
\int_\mathcal{S}\nabla_{(AB}\hat{\breve\xi}_{C)}
\widehat{\nabla^{(AB}\fixedhatbreve{\xi}{}^{C)}} 
\mbox{d}\mu.
\label{AlternativeExpression} 
\end{equation}
Integration by parts and the expressions discussed in the previous paragraphs  imply that:
\begin{align}
M''
={}&
\int_\mathcal{S}(\nabla_{AB}\breve\xi_C+\Omega_{ABCD}\breve\xi^D)(\widehat{\nabla^{(AB}\fixedbreve\xi^{C)}}+\Omega^{ABC}{}_F\hat{\breve\xi}^F) \mbox{d}\mu\nonumber\\
={}&
\int_{\partial\mathcal{S}_\infty}n_{AB}\breve\xi_C
(\widehat{\nabla^{(AB}\fixedbreve\xi^{C)}}+\Omega^{ABC}{}_F\hat{\breve\xi}^F)\mbox{d}S
+\int_\mathcal{S}
\breve\xi_A(\Omega^A{}_{BCD}\widehat{\nabla^{BC}\fixedbreve\xi^D}
-\tfrac{1}{2}\hat{\breve\xi}^A\Omega_{BCDF}\Omega^{BCDF}
) \mbox{d}\mu\nonumber\\
&+\int_\mathcal{S}
\Omega_{ABCD}\breve\xi^D\widehat{\nabla^{(AB}\fixedbreve\xi^{C)}}
+\tfrac{1}{2}\Omega_{BCDF}\Omega^{BCDF}\breve\xi_A\hat{\breve\xi}^A \mbox{d}\mu\nonumber\\
={}&
\int_{\partial\mathcal{S}_\infty}n_{AB}\breve\xi_C
\widehat{\nabla^{(AB}\fixedbreve\xi^{C)}}\mbox{d}S\nonumber\\
={}&
\int_{\partial\mathcal{S}_\infty}n_{AB}(\hat\xi_A+\tfrac{2}{3}\nabla_A{}^B\theta_B+\tfrac{1}{3}K\hat\kappa_A)(
\widehat{\nabla^{(AB}\fixedhat\xi^{C)}}
+\tfrac{2}{3}\widehat{\nabla^{(AB}\nabla^{C)D}\theta_D}
+\tfrac{1}{3}\widehat{\nabla^{(AB}(K\fixedhat\kappa^{C)})})\mbox{d}S\nonumber\\
={}&
\int_{\partial\mathcal{S}_\infty}n_{AB}\hat\xi_A
\widehat{\nabla^{(AB}\fixedhat\xi^{C)}}\mbox{d}S\nonumber\\
={}& M'.
\end{align}
Hence, $M''=M'$ also gives an expression for the total mass.

\subsection{Initial data sets with vanishing mass}
Assume now that the slice $\mathcal{S}$ has the topology of $\mathbb{R}^3$. This implies that it has only one asymptotic end.
Furthermore assume that $M'=0$, that is $m=0$. Then it follows that $M''=0$ and
\[
\nabla_{(AB}\hat\xi_{C)}={}0, \quad 
\nabla_{(AB}\hat{\breve\xi}_{C)}={}0.
\]
Using $\mathbf{L}(\kappa_A)=0$ in equation \eqref{Senxi1a} one obtains
\[
\nabla_{AB}\xi^B=\tfrac{1}{2}K\xi_A \quad \Longleftrightarrow \quad \nabla_{AB}\hat\xi^B=0. 
\]
In the same way we get $\nabla_{AB}\hat{\breve\xi}^B=0$.
These results can be combined to give
\begin{subequations}
\begin{align}
\nabla_{AB}\hat\xi_C=0,\label{xihatconstant}\\
\nabla_{AB}\hat{\breve\xi}_C\label{xibrevehatconstant}=0.
\end{align}
\end{subequations}
Hence, $\xi_C$ and $\breve\xi$ are covariantly constant spinors. One
can exploit this property by taking a derivative of \eqref{xihatconstant}
and \eqref{xibrevehatconstant} and using the commutators. This yields
\begin{subequations}
\begin{eqnarray}
&&0={}-2\nabla_{(A}{}^D\nabla_{B)D}\hat\xi_C
=\Psi_{ABCD}\hat\xi^D,\label{algebraicxihat}\\
&&0={}-2\nabla_{(A}{}^D\nabla_{B)D}\hat\xi_C
=\Psi_{ABCD}\hat{\breve\xi}^D,\label{algebraicxibrevehat}
\end{eqnarray}
\end{subequations}
Consequently, both $\hat\xi_A$ and $\hat{\breve\xi}_A$ are principal spinors of
$\Psi_{ABCD}$. Furthermore,
\[
\hat\xi_A\hat{\breve\xi}^A=\hat\xi_A\hat{\hat\xi}^A+o_\infty(r^{-1/2})=\xi_A\hat\xi^A+o_\infty(r^{-1/2})=1+o_\infty(r^{-1/2}),
\]
and also $\nabla_{BC}(\hat\xi_A\hat{\breve\xi}^A)=0$. So, one can
conclude that $\hat\xi_A\hat{\breve\xi}^A=1$. We can therefore use
$\{\hat\xi_A,\hat{\breve\xi}_A\}$ as a dyad on the entire slice
$\mathcal{S}$. The equations
\eqref{algebraicxihat}-\eqref{algebraicxibrevehat} then yield
$\Psi_{ABCD}=0$ on $\mathcal{S}$.  


\medskip
\noindent
\textbf{Remark 1.}
Given that $\Psi_{ABCD}=0$ on $\mathcal{S}$, one can use known results
on the causal propagation of the Weyl tensor in vacuum spacetimes to
conclude that $\Psi_{ABCD}=0$ on the development of $\mathcal{S}$
---see e.g. \cite{BonSen97}. The asymptotic conditions then imply
that the development is the Minkowski spacetime.  However, it is of
interest to conclude the same result from arguments purely intrinsic
to the hypersurface that imply that the invariant $I$ as given by
equation \eqref{geometric:invariant} has to satisfy $I=0$.

\medskip
To pursue the idea expressed in the previous paragraph we proceed as
follows: From the equivalence
\[
\nabla_{(AB} \hat{\xi}_{C)}=0 \Longleftrightarrow \nabla_{(AB} \xi_{C)} = -\Omega_{ABCQ}\xi^Q,
\]
and equation \eqref{Senxi1b} it follows that
\begin{equation}
2\nabla_{(A}{}^D\xi_{BC)D}
+\Psi_{ABCD}\kappa^D
+\tfrac{2}{3}K\xi_{ABC}
-\xi_{(A}{}^{DF}\Omega_{BC)DF}=0\label{curlxi3}
\end{equation}
Contracting \eqref{curlxi3} with $\hat\xi^C$ and $\hat{\breve\xi}^C$
and using \eqref{algebraicxihat}-\eqref{algebraicxibrevehat} yields
\begin{subequations}
\begin{eqnarray}
&& \hat\xi^C\nabla^D{}_{(A}\xi_{BC)D}={}
\tfrac{1}{3}\hat\xi^C\xi_{(A}{}^{DF}\Omega_{B)CDF}
+\tfrac{1}{6}\hat\xi^C\xi_C{}^{DF}\Omega_{ABDF}
-\tfrac{1}{3}K\xi^C\xi_{ABC},\label{CurlContraction:A}\\
&& \hat{\breve\xi}^C\nabla^D{}_{(A}\xi_{BC)D}={}
\tfrac{1}{3}\hat{\breve\xi}^C\xi_{(A}{}^{DF}\Omega_{B)CDF}
+\tfrac{1}{6}\hat{\breve\xi}^C\xi_C{}^{DF}\Omega_{ABDF}
-\tfrac{1}{3}K\hat{\breve\xi}^C\xi_{ABC}. \label{CurlContraction:B}
\end{eqnarray}
\end{subequations}
Define the spinors
\[
\nu_{AB}\equiv\xi_{ABD}\hat\xi^D, \quad \breve\nu_{AB}\equiv\xi_{ABD}\hat{\breve\xi}^D.
\]
The sharp multiplication result for weighted Sobolev spaces of
Lemma 2.4 in \cite{Max06} yields $\nu_{AB}, \breve\nu_{AB}\in
H^\infty_{-3/2}$. The equations
\eqref{CurlContraction:A}-\eqref{CurlContraction:B} can be reexpressed
as
\begin{eqnarray*}
&& \nabla_{(A}{}^C\nu_{B)C}={}
\tfrac{1}{2}\Omega_{ABCD}\nu^{CD}-\tfrac{1}{3}K\nu_{AB},\\
&& \nabla_{(A}{}^C\breve\nu_{B)C}={}
\tfrac{1}{2}\Omega_{ABCD}\breve\nu^{CD}-\tfrac{1}{3}K\breve\nu_{AB}.
\end{eqnarray*}
Which is equivalent to 
\[
D_{(A}{}^C\nu_{B)C}={}0, \quad D_{(A}{}^C\breve\nu_{B)C}={}0.
\]
From the trivial topology of the 3-dimensional manifold $\mathcal{S}$ one concludes that there exist (globally) scalars $\nu, \breve\nu\in H^\infty_{-1/2}$ such that 
\[
\nu_{AB}=D_{AB}\nu, \quad \breve\nu_{AB}=D_{AB}\breve\nu.
\]
The normalisation condition $\hat\xi_A\hat{\breve\xi}^A=1$ then implies that
\[
\xi_{ABC}=\xi_{ABD}\hat{\breve\xi}^D\hat\xi_C
-\xi_{ABD}\hat\xi^D\hat{\breve\xi}_C=
\breve\nu_{AB}\hat\xi_C
-\nu_{AB}\hat{\breve\xi}_C.
\]
By virtue of the equations \eqref{xihatconstant}-\eqref{xibrevehatconstant}
we have
\[
\xi_{ABC}=\nabla_{(AB}\left(\breve\nu\hat\xi_{C)}-\nu\hat{\breve\xi}_{C)}\right).
\]
Hence, $\mathbf{L}(\breve\nu\hat\xi_C-\nu\hat{\breve\xi}_C)=0$. Note
also that $\breve\nu\hat\xi_C-\nu\hat{\breve\xi}_C\in
H^\infty_{-1/2}$, by the sharp multiplication result. Thus, the triviality of the Kernel of
$\mathbf{L}:H^\infty_{-1/2}\rightarrow H^\infty_{-5/2}$ yields
\[
\breve\nu\hat\xi_C-\nu\hat{\breve\xi}_C=0,
\]
 and, moreover,
$\xi_{ABC}=0$. One concludes that $I=0$. Therefore, the spacetime is isometric to the Minkowski spacetime.

\medskip
\noindent
\textbf{Remark.}
It is of interest to note that the integral $M'$ as given by
expression \eqref{KprimDef} evaluated over a subset
$\mathcal{U}\subset\mathcal{S}$ can be interpreted the mass of that portion of the slice; it is non-negative by construction and it tends to the complete ADM mass as $\mathcal{U}$ grows to cover
$\mathcal{S}$. The drawback of this construction is that one needs to
solve $\mathbf{L}(\kappa_A)$ on the entire slice $\mathcal{S}$ 
---that is, it is a construction that needs global information.

\section{Conclusions}
In this article we have used the idea of quantifying how much a
spacetime fails to have a particular symmetry to construct a global
geometric invariant characterising initial data sets for the Minkowski
spacetime. Not surprisingly, this invariant turns out to be related to
the ADM mass of the data. The approach advocated in this article
provides expressions for the mass as $L^2$ norms of some auxiliary
spinorial fields. This suggests that variations of our approach could
be used to obtain estimates of parts of the initial data in terms of
the mass. A property of potential relevance for the discussion of the
time evolution of the invariants is the ability to switch, according
to need, between expressions for the invariants given in terms of bulk or surface
integrals. As pointed out in the introduction, the main motivation for
the analysis presented in this article is to develop intuition and a
mathematical toolkit for the analysis of similar questions for the
more complicated geometric invariants of
\cite{BaeVal10a,BaeVal10b}. These tantalising possibilities will be
analysed elsewhere.

\medskip

\section*{Acknowledgements}
We thank A Garc\'{\i}a-Parrado for valuable comments.
TB is funded by a scholarship of the Wenner-Gren foundations. JAVK was
funded by an EPSRC Advanced Research fellowship.


\end{document}